\documentclass[letterpaper, 10 pt, journal, twoside]{IEEEtran}

\usepackage{amsmath}
\usepackage{amsthm}
\usepackage{color}
\usepackage{amssymb}
\usepackage[colorlinks=true, urlcolor=blue, citecolor=blue, linkcolor=blue]{hyperref}
\usepackage{amsfonts}
\usepackage{amstext}
\usepackage{psfrag}
\usepackage{graphicx}
\usepackage{enumerate}
\usepackage{balance}
\usepackage{cite}
\usepackage{mathtools}
\usepackage{algorithm,algpseudocode}

\IEEEoverridecommandlockouts                              

\newcounter{Theorem}
\newtheorem{theorem}{Theorem}[Theorem]

\newcounter{Definition}
\newtheorem{definition}{Definition}[Definition]

\newcounter{Assumption}
\newtheorem{assumption}{Assumption}[Assumption]

\newcounter{Problem}
\newtheorem{problem}{Problem}[Problem]

\newcounter{Remark}
\newtheorem{remark}{Remark}[Remark]

\DeclareMathOperator{\He}{He}
\DeclareMathOperator{\maximize}{maximize}
\DeclareMathOperator{\dom}{dom}
\DeclareMathOperator{\R}{\mathbb{R}}
\DeclareMathOperator{\N}{\mathbb{N}}
\newcommand{\EndSym}{\hfill $\bigtriangleup$}
\newcommand{\source}{{This is an archival version of our paper. Please cite the published version DOI:  \href{https://doi.org/10.1109/LCSYS.2019.2911882}{https://doi.org/10.1109/LCSYS.2019.2911882}}}

\makeatletter
\def\ps@IEEEtitlepagestyle{}
\title{
On DoS Resiliency Analysis of Networked Control Systems: Trade-off Between Jamming Actions and Network Delays}

\author{Roberto~Merco$^{\dagger}$,~\IEEEmembership{Student Member,~IEEE,}
        Francesco~Ferrante$^{\ddagger}$,~\IEEEmembership{Member,~IEEE,}
        and~Pierluigi~Pisu$^{\dagger}$,~\IEEEmembership{Member,~IEEE}
\thanks{$^{\dagger}$Automotive Department of Clemson University, Greenville SC 29607 USA. Email: rmerco@clemson.edu, pisup@clemson.edu.}%
\thanks{$^{\ddagger}$Univ. Grenoble Alpes, CNRS, GIPSA-lab, F- 38000 Grenoble, France. Email: francesco.ferrante@gipsa-lab.fr}
\thanks{This material is based upon work supported by the National Science Foundation (NSF) under grant No. CNS-1544910. Any opinions, findings and conclusions or recommendations expressed in this material are those of the authors and do not necessarily reflect the views of the National Science Foundation.}
}

\usepackage{fancyhdr}
\begin{document}

\maketitle

\begin{abstract}
This letter deals with the problem of quantifying resiliency of Networked Control Systems (NCSs) to Denial-of-Service (DoS) attacks and variable network delays. Internal exponential stability and $\mathcal{L}_2$ external stability are studied. The closed-loop system is augmented with an auxiliary timer variable and analyzed in a hybrid system framework. Lyapunov-like conditions are given to ensure $0$-input global exponential stability and $\mathcal{L}_2$ external stability. A computationally affordable algorithm based on linear matrix inequalities is devised to provide trade-off curves between maximum length of DoS attacks and largest network delays. Finally, the effectiveness of the proposed approach is shown in a numerical example.
\end{abstract}


\section{INTRODUCTION}
\IEEEPARstart{I}{n} Networked Control Systems (NCSs) controller, sensors, and actuators can be partially or entirely distributed and connected through a wired or wireless communication network. NCSs have a vast range of applications in mobile robots, intelligent transportation systems, and remote surgery, just to mention a few. Network imperfections, such as sampling, network delays, and packet dropping, affect performance and stability of NCSs and are in general unavoidable due to limited bandwidth, network traffic, and transmission protocols. Due to their relevance in applications, NCSs have seen an increasing interest in the community; see, e.g., \cite{hespanha2007survey,zhang2013network,
zhang2016survey}. Beside network imperfections, cyber attacks exploit the network to deteriorate the performance or induce instability of NCSs \cite{mo2014detecting,teixeira2015secure}. Among cyber attacks, Denial-of-Service (DoS) is the easiest to accomplish and impacts stability and performance of NCSs \cite{yuan2013resilient,amin2009safe}. DoS attacks can induce packet losses and can be easily performed by jamming strategies \cite{poisel2011modern} with the objective of blocking the transmission of information between nodes \cite{teixeira2012attack}. 

Packet losses are often treated with either stochastic or deterministic models. While the former is more suitable to capture natural packet dropping phenomena, the latter performs better in case of packet dropping due to security reasons \cite{amin2009safe}. In fact, a stochastic characterization of packet dropouts generated by cyber attacks would limit the ability of NCSs to capture the malicious and intelligent nature of DoS strategies \cite{dolk2017event}. Furthermore, duration and frequency of DoS attacks are often considered limited over time, since attackers deal with detection avoidance, implementation simplicity, limited resources \cite{amin2009safe,hu2018resilient}, and mitigation techniques \cite{xu2006jamming, debruhl2011digital}. A recent comparison of deterministic packet dropping models can be found in \cite{de2018comparison}.

In the literature, various contributions deal with the design of NCSs by relying on deterministic models of DoS attacks. To cite few recent works, De Persis et al. \cite{de2015input} propose a framework to explicitly characterize the frequency and the duration of DoS attacks under which stability of the closed-loop system is preserved. Dolk et al. \cite{dolk2017event} build upon \cite{de2015input} by considering control systems with output feedback controllers. Feng et al. \cite{feng2017resilient} propose a control scheme equipped with prediction capabilities to reconstruct the missing measurements during DoS attacks and take into account network delays. However, to the best of authors knowledge, network delays are not commonly taken into account when studying the effect of DoS attacks.

In this letter, we analyze the resiliency of NCSs to DoS attacks by using an emulation approach 
\cite{carnevale2007lyapunov}. 
More specifically, we assume that a dynamic output feedback controller has been already designed to guarantee stability of the ``networked-free" closed-loop system and we analyze the impact of its networked implementation. We employ a deterministic model of the packet dropping induced by DoS attacks. In this setting, we provide estimates on the worst-case bounds of the maximum allowable number of successive packet dropouts (MANSD) and network delays (MAD) under which stability is preserved. As in \cite{feng2017resilient}, we consider NCSs where only the sensing path of the closed-loop system is subject to network communication.

Our contribution is as follows. We propose sufficient conditions in the form of matrix inequalities to get explicit bounds for MANSD and MAD that a given NCS can tolerate to maintain closed-loop stability. The approach we pursue relies on Lyapunov theory for hybrid systems in the framework of \cite{goebel2012hybrid}. Specifically, sufficient conditions for global zero-input exponential stability and $\mathcal{L}_2$ external stability are given.
One of the unique features of our approach is that the parameters of interest in the problem, i.e., MANSD and MAD, appear explicitly in the resulting conditions. This enables the derivation of a computationally affordable algorithm for the approximation of the trade-off curve between MANSD and MAD based on semidefinite programming tools. 

The remainder of the letter is organized as follows. Section II presents the modeling of the considered NCS and states the problem we solve. Section III provides sufficient conditions for the solution to the considered problem. Section IV describes an LMI-based algorithm for the estimation of the trade-off curves between MANSD and MAD. The effectiveness of the proposed methodology is shown in Section V through a numerical example. Finally, we provide concluding remarks in Section VI.

\subsubsection*{\textbf{Notation}}
The set $\mathbb{N}$ denotes the set of strictly positive integers, $\mathbb{N}_0=\mathbb{N}\cup\{0\}$, 
$\mathbb{R}^n$ represents a vector of dimension $n$, $\mathbb{R}^{n\times m}$ is the set of $n \times m$ real matrices, and $S_+^n$ represents the set of $n \times n$ symmetric positive definite matrices. 
The identity matrix and the null matrix are denoted, respectively, by $\textbf{I}$ and $\textbf{0}$. Give any $A\in \mathbb{R}^{n \times m}$, $A^\top$ denotes the transpose of $A$, $A^{-\top}=(A^\top)^{-1}$ (when $A$ is nonsingular), $\textrm{He}(A)=A+A^\top$. 
For a symmetric matrix $A$, $A>0$ and $A \geq 0$ ($A<0$ and $A \leq 0$) means that $A$ $(-A)$ is, respectively, positive definite and positive semidefinite, $\lambda_{\min}$ and $\lambda_{\max}$ denote respectively the smallest and the largest eigenvalue of $A$. In partitioned symmetric matrices, the symbol $\bullet$ represents a symmetric block. For a vector $x \in \mathbb{R}^n$, $\vert x \vert$ denotes the Euclidean norm. Given two vectors $x$ and $y$, we denote $(x,y)=[x^\top,y^\top]^\top$. Given a vector $x \in \mathbb{R}^n$ and a closed set $\mathcal{A}$, the distance of $x$ to $\mathcal{A}$ is defined as $\vert x \vert_\mathcal{A}=\textrm{inf}_{y\in\mathcal{A}} \vert x-y \vert$. For any function $z: \mathbb{R} \rightarrow \mathbb{R}^n$, we denote $z(t^+):=\textrm{lim}_{s\rightarrow t^+} z(s)$ when it exists. By $\vee$ and $\wedge$ we denote, respectively, the logical ``or'' and ``and''. A function $\alpha\colon\mathbb{R}_{\geq 0} \rightarrow \mathbb{R}_{\geq 0}$ is said to be of class $\mathcal{K}$ if it is continuous, $\alpha(0)=0$, and it is strictly increasing.

\section{Problem Statement}
\subsection{Description}
We consider NCSs, depicted in Fig. \ref{fig:controlScheme}, where a plant $\mathcal{P}$ is stabilized by a dynamic controller $\mathcal{K}$ that relies on measurements collected through a packet-based network subject to DoS attacks and variable network delays. The presence of the network results into an intermittent stream of information
from the plant to the controller, which does not have access to the plant output in a continuous-time fashion. To overcome this problem, we assume that the controller is equipped with a Zero Order Hold (ZOH) device that stores the last measurement received from the plant and holds it constant until new measurement data are available.

\begin{figure}[thpb]
\centering
\includegraphics[scale=0.7]{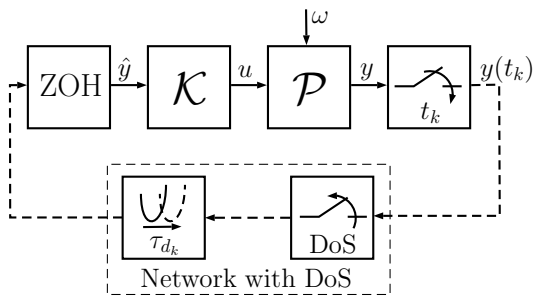}
\caption{Schematic representation of the NCS considered in this letter.}
\label{fig:controlScheme}
\end{figure}

In this letter, we assume the plant is a linear time-invariant continuous-time system of the form:
\begin{equation} \label{eq:plant}
\mathcal{P}: \quad \dot{x}_p = A_p x_p + B_p u + W \omega, \quad y = C_p x_p
\end{equation}
where $x_p \in \mathbb{R}^{n_{x_p}}$ represents the state of the plant, $u \in \mathbb{R}^{n_u}$ represents the control input, $\omega \in \mathbb{R}^{n_\omega}$  is an exogenous disturbance, and $y \in \mathbb{R}^{n_{y}}$ is the output of the plant. The constant matrices $A_p$, $B_p$, $W$, and $C_p$ are given and of appropriate dimensions. We consider a setup in which the measurement $y$ is sampled and transmitted to the controller periodically with a period $T_s$. Specifically, we suppose that the output of the plant is sampled and transmitted at certain time instants $t_k, \, k\in \N_0$ with $t_{k+1}-t_k=T_s$, $t_0=0$.

Inspired by \cite{dolk2017event,feng2017resilient}, we consider a DoS attack as a limited time interval in which a malicious jamming attacker blocks the communication channel. Therefore, DoS attacks can be seen as a sequence of intervals $\{H_n\}_{n\in\mathbb{N}}$ where transmissions of the output $y$ fail to reach the controller for a limited amount of time. Specifically, we assume that the $n$-th DoS attack generates $\iota_n \in \{0,1, \ldots,\Delta\}, \forall n\in\mathbb{N}$, successive packet dropouts, where $\Delta \in \N_0$ is the \textit{maximum allowable number of successive packet dropouts} (MANSD). Furthermore, we assume that intervals $\{H_n\}_{n\in\mathbb{N}}$ do not overlap to each other, and that there is at least one successful transmission in between them. Fig. \ref{fig:evolutionDoS} depicts a graphical illustration of a possible admissible sequence $\{H_n\}_{n\in\mathbb{N}}$. 

\begin{figure}[thpb]
\centering
\includegraphics[scale=0.78]{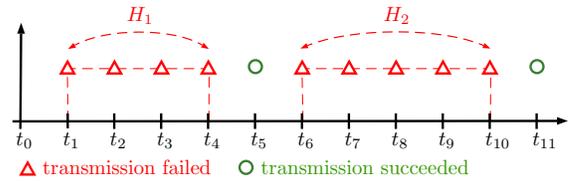}
\caption{Illustration of the evolution of packet dropping due to DoS attacks.}
\label{fig:evolutionDoS}
\end{figure}

\begin{remark}
It is worth mentioning that in \cite{de2018comparison} one can find less conservative models for DoS attacks. In this letter, with the objective of reducing the complexity of the model of the closed-loop system, we consider instead DoS attacks characterized by MANSD.
\end{remark}

 Concerning network delays, in this letter we employ the \textit{small delay assumption}  \cite{heemels2010networked}. In particular, when transmissions are successful, $y$ is received by the controller after a bounded, possibly time varying, network delay $\tau_{d_k}$. We assume that $0 \leq \tau_{d_k} \leq T_{mad}, k \in \N_0$, where $0\leq T_{mad} \leq T_s$ is the \textit{maximum allowable delay} (MAD). 
 
\begin{remark}
The latter condition implies that the transmitted output measurement must be received by the controller before the next measurement is sampled and sent. A similar assumption is considered also in \cite{heemels2010networked,feng2017resilient} just to cite a few. This  assumption prevents from packet disorder phenomena and allows our model to capture packet dropouts and network delays in a simple and unified fashion.
\end{remark}

The dynamics of the ZOH can be modeled as a system with jumps in its state. In particular, let $\hat{y}\in\R^{n_y}$ be the state of the ZOH. Its dynamics can be given as follows for all $k \in \N_0$:
$$
\scalebox{0.9}{
$
\text{ZOH}:  \left\{\begin{array}{ll}
\dot{\hat{y}}(t) = 0 & \forall t \not= t_k+\tau_{d_k} \vee t_k \in \bigcup_{n\in\mathbb{N}} H_n\\
\hat{y}(t^+) = y(t_k) &  \forall t = t_k+\tau_{d_k} \wedge t_k \notin \bigcup_{n\in\mathbb{N}} H_n \\
\end{array} \right.
$}
$$
Notice that $\hat{y}$ is kept constant and it is set to $y(t_k)$ after a network delay $\tau_{d_k}$ only when transmissions are successful. 

Because of the ZOH device, the controller $\mathcal{K}$ is fed with the piecewise constant signal $\hat{y}$, and its continuous-time dynamics are given by:
\begin{equation} \label{eq:controller}
\mathcal{K} : \quad \dot{x}_c = A_c x_c + B_c \hat{y}, \quad u = C_c x_c + D_c \hat{y}
\end{equation}
where $x_c \in \mathbb{R}^{n_{x_c}}$ is the controller state, and $A_c$, $B_c$, $C_c$, and $D_c$ are given constant matrices of appropriate dimensions.

\subsection{Problem Statement}

Given a performance output $y_o := C_o \left( x_p, x_c \right)$ with $y_o \in \mathbb{R}^{n_{y_o}}$, the problem we solve is as follows:

\begin{problem} \label{prob:stability}
Given plant $\mathcal{P}$, controller $\mathcal{K}$, and sampling time $T_s$, determine the largest achievable values of $T_{mad}$ and $\Delta$ such that the interconnection of plant \eqref{eq:plant} and controller \eqref{eq:controller} satisfies the following properties:
\begin{enumerate}[({P}1)]
\item global exponential stability when the input $\omega$ is identically zero;
\item when the disturbance $\omega$ is bounded, the closed-loop state is bounded;
\item $\mathcal{L}_2$ stability from the disturbance $\omega$ to the performance output $y_o$ is ensured with a desired $\mathcal{L}_2$-gain $\gamma$. \EndSym
\end{enumerate} 
\end{problem}

\subsection{Hybrid Modeling}
The closed-loop system in Fig. \ref{fig:controlScheme} can be modeled as a linear system with jumps in $\hat{y}$. In particular, for all $k \in \N_0$ one obtains
$$
\scalebox{0.9}{
$
\begin{array}{ll}
\left\{ 
\begin{aligned}
& \dot{x}_p = A_p x_p \! + \! B_p C_c x_c \! + \! B_p D_c \hat{y} \! + \! W \omega \\
& \dot{x}_c = A_c x_c \!+\! B_c \hat{y}\\
& \dot{\hat{y}} = 0
\end{aligned}\right. & 
\begin{aligned}
& \forall t \not= t_k+\tau_{d_k} \vee \\
& t_k \in \bigcup_{n\in\mathbb{N}} H_n  \\
\end{aligned}\\[0.7cm]
\left\{ 
\begin{aligned}
& x_p(t^+) = x_p(t)\\
& x_c(t^+) = x_c(t)\\
& \hat{y}(t^+) = C_p x_p(t_k)
\end{aligned}\right. & 
\begin{aligned}
& \forall t = t_k+\tau_{d_k} \wedge \\
& t_k \notin \bigcup_{n\in\mathbb{N}} H_n  \\
\end{aligned}\\
\end{array}    
$}
$$
where, for the sake of notation, dependence on time is omitted in continuous-time dynamics. The closed-loop system evolves with differential equations and experiences jumps. For such a reason, we model it into the hybrid system framework in  \cite{goebel2012hybrid}. To this end, we introduce the auxiliary variables $\tau \in \mathbb{R}_{\geq 0}$, $s_y \in \mathbb{R}^{n_y}$, and $l\in\{0,1\}$. Variable $\tau$ is a timer that keeps track of the duration of sampling intervals and network delays, and  triggers a jump whenever a new measure that will be successfully received by the controller is sampled (sampling events) or received (update events). Variable $s_y$ represents a memory state that stores the value of the sampled measurement $y(t_k)$. In particular, $y(t_k)$ is stored in $s_y$ at sampling events of successful transmissions, and $s_y$ is assigned to $\hat{y}$ at update events after network delays $\tau_{d_k}$. Similarly as in \cite{dolk2017event}, variable $l$ allows one to model both sampling events (when $l=0$) and updating events (when $l=1$). In particular, we consider the following hybrid model of the closed-loop system
\begin{equation}
\scalebox{0.9}{
$
\begin{aligned} \label{eq:firstModel}
\left\lbrace
\begin{array}{ll}
\dot{\xi} =f_{\xi}(\xi,\omega) & \xi \in C_{\xi}, \omega \in \mathbb{R}^{n_\omega} \\
\xi^+ = g_{\xi}(\xi) & \xi \in D_{\xi} \\
y_o = C_o \left( x_p, x_c \right)
\end{array}
\right. 
\end{aligned}
$}
\end{equation}
where $\xi \coloneqq  (x_p,x_c,\hat{y},s_y,\tau,l) \in \mathbb{R}^{n_{\xi}}$ with $n_{\xi}:=n_{x_p}+n_{x_c}+2n_y+2$ is the state of the hybrid system, $f_{\xi}(\xi,\omega) \! \coloneqq \! ( A_p x_p \! + \! B_p C_c x_c \! + \! B_p D_c \hat{y} \! + \! W \omega, A_c x_c \!+\! B_c \hat{y},0,0,1,0)$ and $g_{\xi}(\xi) \! \coloneqq \! (x_p, x_c, (1\!-\!l) \hat{y} \! + \! l s_y, (1\!-\!l) C_p x_p \!+\! l s_y, l\tau,1\!-\!l)$. 
The flow set $C_{\xi}$ and the jump set $D_{\xi}$ are respectively defined by  $C_{\xi} \coloneqq  \{ \xi \in \mathbb{R}^{n_{\xi}} | (l = 0 \wedge \tau \in [0,(\Delta+1)T_s]) \vee (l = 1 \wedge \tau \in [0,T_{mad}]) \}$ and $D_{\xi} \coloneqq \{ \xi \in \mathbb{R}^{n_{\xi}} | (l = 0 \wedge \tau \in T_s \Theta_\Delta) \vee (l = 1 \wedge \tau \in [0,T_{mad}]) \}$
where $\Theta_\Delta \coloneqq \{1,2, \ldots,\Delta+1 \}$. 

\begin{remark}
It is worth noticing that the model in \eqref{eq:firstModel} considers only sampling events for successful transmissions. Those occur periodically with a period that is multiple of $T_s$. In particular, sampling events are triggered for $\tau=T_s$ when no DoS occurs, whereas they are triggered for $\tau\in\iota T_s$, with $\iota:=\{2, \ldots,\Delta+1 \}$, for DoS attacks generating a number of consecutive packet dropouts within $1$ and $\Delta$. This feature is modeled by the condition $(l = 0 \wedge \tau \in T_s \Theta_\Delta)$ in the definition of $D_{\xi}$. Notice that, by definition, $C_\xi$ and $D_\xi$ overlap each other, and, when the state $\xi$ belongs to $C_\xi \cap D_\xi$, both flowing and jumping are allowed. As such, solutions to \eqref{eq:firstModel} are not unique. This enables one to capture all possible network behaviors in a unified manner.
\end{remark}

At this stage, to simplify the analysis, we introduce the following change of coordinates 
\begin{equation} \label{eq:changeCoordinates}
\eta\coloneqq \hat{y}-y, \, \sigma\coloneqq s_y - y 
\end{equation}

Using \eqref{eq:changeCoordinates}, one can obtain, by straightforward calculation, the closed-loop hybrid system in the new coordinates which reads as follows:
\begin{equation} \label{eq:systemHy}
\scalebox{0.9}{
$
\mathcal{H}_y
\begin{aligned}
\left\lbrace
\begin{array}{ll}
\dot{x} = f(x,\omega) & x \in C, \omega \in \mathbb{R}^{n_\omega} \\
x^+  = g(x) & x \in D \\
y_o = C_o \tilde{x}
\end{array}
\right. 
\end{aligned}
$}
\end{equation}
where $x\coloneqq(\tilde{x},\eta,\sigma,\tau,l) \in \mathbb{R}^{n_x}$ is the state with $n_x \coloneqq n_{\tilde{x}}+2n_y+2$, and $\tilde{x}\coloneqq (x_p,x_c)$. The flow map is given by
\begin{equation}
\scalebox{0.9}{
$
\begin{aligned} \label{eq:flowMapHy}
f(x,\omega) \coloneqq & ( A_{xx} \tilde{x} + A_{x\eta} \eta +  A_{x\omega} \omega , A_{\eta x} \tilde{x} + A_{\eta \eta} \eta + A_{\eta \omega} \omega , \\
 & A_{\eta x} \tilde{x} + A_{\eta \eta} \eta + A_{\eta \omega} \omega,1,0) \quad \forall x \in C, \omega \in \mathbb{R}^{n_\omega}
\end{aligned}
$}
\end{equation}
where $A_{\eta x} \!\! \coloneqq \!\! \left[ -C_p \,\, \textbf{0} \right] A_{xx}$,  $A_{\eta \eta} \!\! \coloneqq \!\! - C_p B_p D_c$,  $A_{\eta \omega} \! \coloneqq \!  - C_p W$, and
$$
\scalebox{0.9}{
$
A_{xx}\! \coloneqq \!\left[ \!\!\! \begin{array}{cc}
A_p\!\!+\!\!B_p D_c C_p \!\!&\!\!B_p C_c \\
B_c C_p \!\!&\!\! A_c
\end{array} \!\!\! \right]\!\!, \, \! A_{x\eta} \! \coloneqq \! \left[ \!\!\! \begin{array}{c}
B_p D_c \\
B_c
\end{array} \!\!\! \right] \!\!, \, \! A_{x\omega} \! \coloneqq \! \left[ \!\!\! \begin{array}{c}
W \\ 
\textbf{0}
\end{array} \!\!\! \right]\\
$}
$$
derive from \eqref{eq:plant}, \eqref{eq:controller} and \eqref{eq:changeCoordinates}. The jump map is defined for all $x \in D$ by $g(x) \coloneqq (\tilde{x}, l \sigma + (1-l) \eta,l\sigma, l\tau,1-l)$.
The flow set $C$ and the jump set $D$ are respectively defined as $C \coloneqq  \{ x \in \mathbb{R}^{n_x} | (l = 0 \wedge \tau \in [0,(\Delta+1)T_s]) \vee (l = 1 \wedge \tau \in [0,T_{mad}]) \}$ and $D \coloneqq \{ x \in \mathbb{R}^{n_x} | (l = 0 \wedge  \tau \in T_s \Theta_\Delta) \vee (l = 1 \wedge \tau \in [0,T_{mad}]) \}$.

Concerning the existence of Zeno solutions, by analyzing the hybrid domain of solutions to $\mathcal{H}_y$, one can conclude that any maximal solution $\phi$ to $\mathcal{H}_y$ can experience at most three consecutive jumps without flowing in an interval of length smaller than $T_s$. In particular, for every $(t, j) \in \text{dom}\, \phi$, one has $j \leq \frac{t}{T_s}+3$. The latter property is commonly denoted as average dwell time \cite[Example 2.15]{goebel2012hybrid}, and
rules out the existence of Zeno solutions.

\section{Main Results}  \label{sec:stability} 
To solve Problem \ref{prob:stability}, our approach consists of deriving sufficient conditions in the form of matrix inequalities ensuring that the following set 
\begin{equation} \label{eq:setA}
\scalebox{0.9}{
$
\mathcal{A} \coloneqq \{0\} \times \{0\} \times \{0\} \times [0,(\Delta+1)T_s] \times \{0,1\}
$}
\end{equation}
is exponentially stable whenever $\omega \equiv 0$ and, in case of nonzero disturbance $\omega$, the hybrid system $\mathcal{H}_y$ in \eqref{eq:systemHy} is input-to-state stable with respect to $\mathcal{A}$. 

The following definition will be considered in the paper:

\begin{definition}[Exponential input-to-state stability]
\label{defISS}
Let $\mathcal{A}\subset\R^{n_x}$ be closed. System \eqref{eq:systemHy} is \emph{exponentially input-to-state-stable} (eISS) with respect to $\mathcal{A}$ if there exist $\kappa, \lambda>0$, and $p\in\mathcal{K}$ such that each maximal solution pair\footnote{A pair $(\phi,\omega)$ is a solution pair to $\mathcal{H}_y$ if it satisfies its dynamics; see \cite{cai2009characterizations} for more details.} $(\phi,\omega)$ to \eqref{eq:systemHy} is complete, and, if $\Vert \omega \Vert_\infty$ is finite, it satisfies
\begin{equation}
\label{eq:defeISS}
\scalebox{0.9}{
$
\vert \phi(t,j)\vert_{\mathcal{A}}\leq \max\{\kappa e^{-\lambda (t+j)}\vert \phi(0,0)\vert_{\mathcal{A}}, p(\Vert \omega\Vert_{\infty})\}
$}
\end{equation}
for each $(t,j)\in\dom\phi$, where $\Vert \omega \Vert_\infty$ denotes the $\mathcal{L}_\infty$ norm of the hybrid signal $\omega$ as defined in \cite{nevsic2013finite}.
\EndSym 
\end{definition}

At this stage, consider the following assumption. 

\begin{assumption} \label{ass:lyapunov}
Let $\gamma$ be a given positive real number. There exist three continuously differentiable functions
$V_1 : \mathbb{R}^{n_{\tilde{x}}} \rightarrow  \mathbb{R}$, $V_2 : \mathbb{R}^{n_y+2} \rightarrow  \mathbb{R}$, $V_3 : \mathbb{R}^{n_y+2} \rightarrow  \mathbb{R}$ and positive real numbers $\alpha_1$, $\alpha_2$, $\beta_1$, $\beta_2$, $\theta_1$, $\theta_2$, and $\lambda_t$,  such that
\begin{enumerate}[({A}1)]
\item $\alpha_1 |\tilde{x}|^2 \leq V_1(\tilde{x}) \leq \alpha_2 |\tilde{x}|^2, \quad \forall x \in C$ \label{ass:lyapunov:A1}
\item $\beta_1 |\eta|^2 \leq V_2(\eta,\tau,l) \leq \beta_2 |\eta|^2, \quad \forall x \in C$
\label{ass:lyapunov:A2}
\item $\theta_1 |\sigma|^2 \leq V_3(\sigma,\tau,l) \leq \theta_2 |\sigma|^2, \quad \forall x \in C$
\label{ass:lyapunov:A3}
\item $V_2(\eta,0,1) + V_3(0,0,1) \leq V_2(\eta,\tau,0) + V_3(\sigma,\tau,0)$, $\forall \eta \in \R^{n_y}, \, \sigma \in \R^{n_y}, \, \tau \in T_s \Theta_\Delta$ \label{ass1:jump1}
\item $V_2(\sigma,\tau,0) + V_3(\sigma,\tau,0) \leq V_2(\eta,\tau,1) + V_3(\sigma,\tau,1)$, $\forall \eta \in \R^{n_y}, \, \sigma \in \R^{n_y},\, \tau \in [0,T_{mad}]$ \label{ass1:jump2}
\item the function $x \mapsto V(x) := V_1(\tilde{x})+V_2(\eta,\tau,l)+V_3(\sigma,\tau,l)$ satisfies  $\langle \nabla V(x), f(x,\omega) \rangle \leq -2\lambda_t V(x) - \tilde{x}^\top C_o^\top C_o \tilde{x} + \gamma^2 \omega^\top \omega$ for each $x \in C, \omega \in \mathbb{R}^{n_\omega}$ \label{ass1:Vdot} \EndSym
\end{enumerate} 
\end{assumption}

The result given next provides sufficient conditions for the solution to Problem~\ref{prob:stability}.

\begin{theorem} \label{th:suffCond}
Let Assumption~\ref{ass:lyapunov} hold. Then:
\begin{enumerate}[($i$)]
\item The hybrid system $\mathcal{H}_y$ is eISS with respect to $\mathcal{A}$;
\item There exists $\alpha>0$ such that any solution pair $(\phi,\omega)$ to $\mathcal{H}_y$ satisfies
$$\scalebox{0.9}{$
\!\!\!\!\sqrt{\!\int_\mathcal{I}\!\!|y_o(r,j(r))|^2 dr}\!\leq\!\alpha |\phi(0,0)|_\mathcal{A} + \gamma \sqrt{\!\int_\mathcal{I}\!\!|\omega(r,j(r))|^2 dr}$}$$
where $\mathcal{I}\coloneqq[0,\sup_t \dom\phi] \cap \dom_t \phi$.
\end{enumerate}
\end{theorem}

\begin{proof}
Consider the following Lyapunov function candidate $V(x) := V_1(\tilde{x})+V_2(\eta,\tau,l)+V_3(\sigma,\tau,l)$ for the hybrid system \eqref{eq:systemHy} defined for every $x \in \mathbb{R}^{n_x}$. We prove (i) first. By setting
$\rho_1 = \min\{\alpha_1,\beta_1,\theta_1\}$, $\rho_2 = \max\{\alpha_2,\beta_2,\theta_2\}$ and in view of the definition of the set $\mathcal{A}$ in (\ref{eq:setA}) one gets 
\begin{equation} \label{eq:proofTheoremLyapBounds}
\scalebox{0.9}{
$
\rho_1 |x|_\mathcal{A}^2 \leq V(x) \leq \rho_2 |x|_\mathcal{A}^2 \qquad \forall x \in C \, \cup \,  D
$}
\end{equation}
Moreover, from Assumption~\ref{ass:lyapunov} item (A\ref{ass1:Vdot}) one has that $\forall x \in C,\omega \in \mathbb{R}^{n_\omega}$
\begin{equation} \label{eq:proofTheoremLyap}
\scalebox{0.9}{
$
\langle \nabla V(x), f(x,\omega) \rangle \leq  -2\lambda_t V(x) + \gamma^2 \omega^\top \omega
$}
\end{equation}
and from Assumption~\ref{ass:lyapunov} items (A\ref{ass1:jump1}) and (A\ref{ass1:jump2}), one has that for all $x \in D$
\begin{equation} \label{eq:proofJumps}
\scalebox{0.9}{
$
V(g(x))\leq V(x)
$}
\end{equation}
Let $(\phi,\omega)$ be a maximal solution pair to \eqref{eq:systemHy}. 
Following the same steps as in \cite[proof of Theorem 1]{ ferrante2015hybrid}, using \eqref{eq:proofTheoremLyapBounds}, \eqref{eq:proofTheoremLyap} and \eqref{eq:proofJumps}, for all $(t,j)\in \textrm{dom}\phi$ one has 
$$\scalebox{0.9}{
$
|\phi(t,j)|_\mathcal{A} \leq \textrm{max} \left\lbrace2\sqrt{\frac{\rho_2}{\rho_1}}e^{-\lambda_t t} |\phi(0,0)|_\mathcal{A} , \frac{2\gamma}{\sqrt{2 \lambda_t \rho_1}} \Vert\omega\Vert_\infty \right\rbrace
$}$$
This shows that \eqref{eq:defeISS} holds with $\kappa = 2\sqrt{\rho_2/\rho_1}$, $\lambda = \lambda_t$ and $r \mapsto p(r) := (2\gamma/\sqrt{2 \lambda_t \rho_1}) r$. Hence, since every maximal solution pair to $\mathcal{H}_y$ is complete, ($i$) is established.
To conclude, let $(\phi, \omega)$ be a maximal solution pair to
$\mathcal{H}_y$ and pick $t>0$. Again, by using the same approach as in \cite[proof of Theorem~1]{ ferrante2015hybrid}, thanks to Assumption~\ref{ass:lyapunov} items (A\ref{ass1:jump1}) and (A\ref{ass1:jump2}), since $V$ is nonincreasing at jumps,  one gets {\small $\int_{\mathcal{I}(t)} \tilde{x}(r,j(r))^\top C_o^\top C_o \tilde{x}(r,j(r)) dr  \leq  V(\phi(0,0)) + \gamma^2 \int_{\mathcal{I}(t)} |\omega(r,j(r))|^2 dr$} where ${\mathcal{I}(t)} := [0,t] \cap \textrm{dom}_t \phi$.
Therefore, by taking the limit for $t$ approaching $\textrm{sup}_t \, \textrm{dom}\phi$, thanks to \eqref{eq:proofTheoremLyapBounds}, one gets ($ii$) with $\alpha=\rho_2$. Hence, the result is established.
\end{proof}

\subsection{Construction of the Lyapunov Function}
To determine the values of $\Delta$ and $T_{mad}$, one needs to explicitely identify functions $V_1$, $V_2$ and $V_3$ in Assumption~\ref{ass:lyapunov}. Let $P_1 \in \mathcal{S}^{n_{\tilde{x}}}_+$, $P_{2,l}:= (1-l)P_{2,0}+l P_{2,1}, \, P_{3,l}:= (1-l)P_{3,0}+l P_{3,1}$ with $P_{2,0}$, $P_{2,1}$, $P_{3,0}$, $P_{3,1} \in \mathcal{S}^{n_y}_+$, $l\in\{0,1\}$, and $\delta$ be positive real number. Inspired by \cite{ferrante2015hybrid} we operate the following selection:
\begin{equation} \label{eq:lypFunctions}
\scalebox{0.9}{
$
\begin{aligned}
V_1(\tilde{x}) = & \tilde{x}^\top P_1 \tilde{x}, \quad V_2(\eta,\tau,l) = e^{-\delta \tau} \eta^\top P_{2,l} \eta, \\
& V_3(\sigma,\tau,l) = e^{-\delta \tau} \sigma^\top P_{3,l} \sigma \\ 
\end{aligned}
$}
\end{equation}
By exploiting the (quasi)-quadratic nature of the Lyapunov function candidate $x \mapsto V_1(\tilde{x})+ V_2(\eta,\tau,l)+V_3(\sigma,\tau,l)$, such a choice for $V_1$, $V_2$, and $V_3$ allows us to cast the solution to Problem \ref{prob:stability} as a solution to some matrix inequalities.

\begin{theorem} \label{th:practical}
If there exist  $P_1 \in \mathcal{S}^{n_{\tilde{x}}}_+$, $P_{2,0},\,P_{2,1},\,P_{3,0}$, and $P_{3,1} \in \mathcal{S}^{n_y}_+$ with
\begin{subequations} \label{eq:Pconds}
\begin{align} 
\scalebox{0.9}{$P_{2,1} - e^{-\delta (\Delta+1)T_s} P_{2,0} \leq 0 \label{eq:P2cond}$} \\
\scalebox{0.9}{$P_{2,0} + P_{3,0} - P_{3,1} \leq 0 \label{eq:P3cond}$} 
\end{align}
\end{subequations}
such that
\begin{equation}
\label{eq:Mconditions}
\scalebox{0.9}{
$
\begin{array}{cc}
\mathcal{M}(0,1)<0&\mathcal{M}(T_{mad},1)<0\\
\mathcal{M}(T_{mad},0)<0& 
\mathcal{M}((\Delta+1)T_s,0)<0
\end{array}
$}
\end{equation}
where for each $\tau \in [0,(\Delta+1)T_s]$ and $l \in \{0,1\}$ the function $\mathcal{M}\colon(\tau, l)\mapsto \mathcal{M}(\tau,l)$ is defined in \eqref{LMI} (at the top of the next page), then the Assumption~\ref{ass:lyapunov} holds.
\end{theorem}
\begin{figure*}[!t]
\vspace*{5pt}
\normalsize
\setcounter{equation}{14}
\begin{small}
\begin{equation}
\label{LMI} 
\scalebox{0.9}{
$\mathcal{M}(\tau,l)= \left( 
\begin{array}{cccc}
\textrm{He}(P_1 A_{xx})+ C_o^\top C_o &  P_1 A_{x\eta} + e^{-\delta \tau} A_{\eta x}^\top P_{2,l} & e^{-\delta \tau} A_{\eta x}^\top P_{3,l} & P_1 A_{x\omega}\\
\bullet & e^{-\delta \tau} \left(\He(P_{2,l} A_{\eta\eta})- \delta  P_{2,l} \right) & e^{-\delta \tau} A_{\eta \eta}^\top P_{3,l}  & e^{-\delta \tau} P_{2,l} A_{\eta \omega} \\
\bullet & \bullet  &  -\delta e^{-\delta \tau} P_{3,l}  & e^{-\delta \tau} P_{3,l} A_{\eta \omega} \\
\bullet & \bullet  &  \bullet & -\gamma^2 \textbf{I}\\
\end{array}
\right)$}
\end{equation}
\end{small}
\setcounter{equation}{15}
\vspace*{-20pt}
\end{figure*}
\begin{proof}
Let $V_1$, $V_2$ and $V_3$ be defined as in (\ref{eq:lypFunctions}). By selecting 
$\scalebox{0.9}{$\alpha_1=\lambda_{\min}(P_1), \, \alpha_2=\lambda_{\max}(P_1), \, \beta_1=\textrm{min}\{ \lambda_{\min}(P_{2,0})e^{-\delta (\Delta+1)T_s},$}$
$\scalebox{0.9}{$\lambda_{\min}(P_{2,1})e^{-\delta T_{mad}} \}$}$, 
$\scalebox{0.9}{$\beta_2=\textrm{max}\left\lbrace \lambda_{\max}(P_{2,0}), \lambda_{\max}(P_{2,1}) \right\rbrace,$}$ 
$\scalebox{0.9}{$\theta_1=\textrm{min}\left\lbrace \lambda_{\min}(P_{3,0})e^{-\delta (\Delta+1)T_s}, \lambda_{\min}(P_{3,1})e^{-\delta T_{mad}} \right\rbrace$}$, and
$\scalebox{0.9}{$\theta_2=\textrm{max}\left\lbrace \lambda_{\max}(P_{3,0}), \lambda_{\max}(P_{3,1}) \right\rbrace$}$,
items (A\ref{ass:lyapunov:A1}), (A\ref{ass:lyapunov:A2}) and (A\ref{ass:lyapunov:A3}) of the Assumption~\ref{ass:lyapunov} are satisfied. By using \eqref{eq:P2cond}, and \eqref{eq:P3cond} one can show that items (A\ref{ass1:jump1}) and (A\ref{ass1:jump2}) hold. Regarding item (A\ref{ass1:jump1}), one has that, for all $x\in D$ with $l=0$,
$ 
V_2(\eta,0,1) + V_3(0,0,1) - V_2(\eta,\tau,0) - V_3(\sigma,\tau,0) \leq \eta^\top \left(P_{2,1}  - e^{-\delta \tau}P_{2,0}   \right) \eta 
$.
Notice that, since $\delta>0$, for all $\tau \in  T_s \Theta_\Delta$,
$
P_{2,1} - e^{-\delta \tau}P_{2,0} \leq P_{2,1} - e^{-\delta (\Delta+1)T_s} P_{2,0}   
$.
Therefore, the satisfaction of \eqref{eq:P2cond} implies $P_{2,1} - e^{-\delta \tau}P_{2,0} \leq 0$ for all $\tau \in  T_s \Theta_\Delta$, which shows that (A\ref{ass1:jump1}) holds.
For item (A\ref{ass1:jump2}), one has that for all $x \in D$ with $l=1$ due to \eqref{eq:P3cond}, it follows that
$
V_2(\sigma,\tau,0) + V_3(\sigma,\tau,0) - V_2(\eta,\tau,1) - V_3(\sigma,\tau,1) \leq e^{-\delta \tau} \sigma^\top \left( P_{2,0} + P_{3,0} - P_{3,1} \right) \sigma \leq 0
$,
which shows that (A\ref{ass1:jump2}) holds because of \eqref{eq:P3cond}.
Regarding item (A\ref{ass1:Vdot}) of Assumption~\ref{ass:lyapunov}, let $V(x)=V_1(\tilde{x})+V_2(\eta,\tau,l)+V_3(\sigma,\tau,l)$. Then, from the definition of the flow map in \eqref{eq:flowMapHy}, for each $x\in C$, $\omega \in \mathbb{R}^{n_\omega}$ one can define $\Omega(x,\omega)\coloneqq\langle \nabla V(x), f(x,\omega) \rangle+\tilde{x}^\top C_o^\top C_o \tilde{x} - \gamma^2 \omega^\top \omega$. Therefore, by defining $\Psi(x,\omega):=(\tilde{x},\eta,\sigma,\omega)$, for each $x \in C$ and $\omega \in \mathbb{R}^{n_\omega}$, one has $\Omega(x,\omega) = \Psi(x,\omega)^\top \mathcal{M}(\tau,l) \Psi(x,\omega)$ where the symmetric matrix $\mathcal{M}$ is given in \eqref{LMI}. Furthermore, notice that it is straightforward to show that there exists $\hat{\lambda}: [0,\tau] \mapsto [0,1]$ such that for each $\tau \in [0,T_{mad}]$, $\mathcal{M}(\tau,1)=\hat{\lambda}(\tau)\mathcal{M}(0,1)+(1-\hat{\lambda}(\tau))\mathcal{M}(T_{mad},1)$ and there exists $\bar{\lambda}: [0,\tau] \mapsto [0,1]$ such that for each $\tau \in [T_{mad},(\Delta+1)T_s]$, $\mathcal{M}(\tau,0)=\bar{\lambda}(\tau)\mathcal{M}(T_{mad},0)+(1-\bar{\lambda}(\tau))\mathcal{M}((\Delta+1)T_s,0)$; see \cite{ferrante2015hybrid}. Therefore, one has that the satisfaction of \eqref{eq:Mconditions} implies $\mathcal{M}(\tau,l)<0, \, \forall(\tau, l)\in [0,(\Delta+1)T_s]\times\{0,1\}$, which lead to $\overline{\varsigma}\coloneqq \max_{(\tau, l)\in[0,(\Delta+1)T_s]\times\{0,1\}}\lambda_{\max}(\mathcal{M}(\tau,l))<0$. Observe that the above quantity is well defined, $(\tau,l)\mapsto \mathcal{M}(\tau,l)$ being continuous on $[0,(\Delta+1)T_s]\times\{0,1\}$. Therefore, one has that for all $x\in C,\omega\in\mathbb{R}^{n_\omega}$, $\Omega(x,\omega)\leq -\overline{\varsigma} \tilde{x}^\top \tilde{x}=-\overline{\varsigma}\vert x\vert^2_\mathcal{A} $.
Defining $\rho_2 = \max\{\alpha_2,\beta_2,\theta_2\}$, using \eqref{eq:proofTheoremLyapBounds} and the definition of $\Omega$, one finally gets, for all $x\in C,\omega\in\mathbb{R}^{n_\omega}$, $\langle \nabla V(x), f(x,\omega) \rangle\leq -\frac{\overline{\varsigma}}{\rho_2}V(x)-\tilde{x}^\top C_o^\top C_o \tilde{x}+\gamma^2 \omega^\top \omega$ which reads as (A\ref{ass1:Vdot}). Hence, the result is established.
\end{proof}

\begin{remark} \label{rem:simplyProb}
Conditions \eqref{eq:Pconds} and \eqref{eq:Mconditions} with $\mathcal{M}(\tau,l)$ without the forth row and the forth column, and $C_o=\textbf{0}$ are sufficient conditions to solve Problem \ref{prob:stability} when only the internal exponential stability ($\omega\equiv 0$) is considered. 
\end{remark}

\section{LMI-Based Algorithm for Trade-off Curves}
In the previous section, we provided sufficient conditions for the solution to Problem~\ref{prob:stability} in the form of matrix inequalities. The objective of the current section is to make use of the proposed sufficient conditions to include some optimization aspects in the solution to Problem~\ref{prob:stability}. In particular, as the primary goal of this letter is to provide estimates of the largest allowable number of consecutive jammed transmissions and communication delay, next we illustrate an algorithm for the approximation of the trade-off curve between this two objectives. Specifically, to accomplish this goal, we make use of conditions \eqref{eq:Pconds} and \eqref{eq:Mconditions} to formulate the following optimization problem:

\begin{equation} \label{eq:optProb}
\scalebox{0.9}{
$
\begin{array}{cl}
\underset{P_1,P_{2,0},P_{2,1},P_{3,0},P_{3,1},\delta,\gamma}{\maximize}
& (\Delta, T_{mad})\\
\text{subject to}
& \eqref{eq:Pconds}, \eqref{eq:Mconditions}, T_{mad} \in [0, T_s] \\
\end{array}
$
}
\end{equation}
where the above maximization is  intended in a Pareto sense \cite{boyd2004convex}. In particular, we provide a systematic approach to build an approximation of the trade-off curve of the above multi-objective optimization problem. Notice that we assume the value of $\gamma$ is given and strictly positive.
To obtain a numerically efficient solution to \eqref{eq:optProb}, we make use of semidefinite programming tools. Specifically, observe that when $\delta$, $\Delta$, and $T_{mad}$ are fixed, conditions \eqref{eq:Pconds} and \eqref{eq:Mconditions} turn into LMIs, which can be efficiently solved via available semidefinite programming solvers \cite{boyd1994linear}. Therefore, checking the feasibility of \eqref{eq:Pconds} and \eqref{eq:Mconditions} can be used in a numerical scheme by performing line search for the scalars $\delta$, $\Delta$, and $T_{mad}$. 

Algorithm 1 describes our approach to solve the optimization problem \eqref{eq:optProb}. The algorithm gives as output the vectors $\vec{\Delta}$ and $\vec{T}_{mad}$ that provide an estimation of the trade-off curve between $\Delta$ and $T_{mad}$. Indeed, the trade-off curve identifies the maximum allowable delay for each given value of consecutive jammed transmissions $\Delta$. In particular, given the value of $\Delta$, one can identify the maximum allowable delay that can affect the last received packet without compromising the stability of the NCS.

\begin{algorithm}[H] 
\caption{Trade-off curve approximation}
\begin{algorithmic}[1] 
\State Initialize $\Delta\gets 0$, $stop \gets 0$, $\vec{\Delta}$, $\vec{T}_{mad}$;
\While{$stop=0$}
	\State maximize $T_d$ on $[0, T_s]$ subject to \eqref{eq:Pconds}-\eqref{eq:Mconditions} by performing line search on $\delta$ for given value of $\Delta$.  
	\If{$T_d$ is found}
		\State $\vec{\Delta}[\Delta] \gets \Delta$, $\vec{T}_{mad}[\Delta] \gets T_d$, $\Delta \gets \Delta+1$
	\Else 
	 	\State $stop \gets 1$
	\EndIf
\EndWhile
\State \Return $\vec{T}_{mad}$, $\vec{\Delta}$
\end{algorithmic}
\end{algorithm}

\begin{remark} 
Let us remark that $T_s$ is a given parameter in the optimization problem \eqref{eq:optProb}. The choice of $T_s$ influences the outcome of Algorithm~1. Indeed, as the sampling time decreases, the closed-loop control system tolerates a larger value of $\Delta$. This is due to the fact that our approach aims at estimating the largest interval $(1+\Delta)T_s$ in which the system can evolve without network updates. Furthermore, notice that $T_{mad}$ is also influenced by $T_s$ since the assumption $0 \leq T_{mad} \leq T_s$. 
\end{remark}

\section{Numerical Example}
In this section, we show-case Algorithm 1 in a specific example.
All numerical results are obtained through the solver \textit{SEDUMI} \cite{sturm1999using} and coded in Matlab$^{\tiny{\textregistered}}$ via \textit{YALMIP} \cite{lofberg2004yalmip}.
We consider the well known batch reactor controlled by a dynamic output feedback controller, presented, e.g., in \cite{nesic2004input,heemels2010networked} and many others. Numerical values of the plant and controller can be found in \cite{nesic2004input} while 
$
\scalebox{0.9}{
$
C_o \!\! =\!\! \left[ \!\! \begin{array}{cccc}
C_p &\!\!\! \textbf{0}\\
\end{array} \!\!\right]
$}
$ 
and 
\scalebox{0.8}{
$W\!\! =\!\! \left[\!\! \begin{array}{cccc}
10 &\!\!\! 0 &\!\!\! 10 &\!\!\! 0 \\
0 &\!\!\! 5 &\!\!\! 0 &\!\!\! 5
\end{array}\!\! \right]^\top \!\!\! $}
in \cite{heemels2010networked}. By employing Algorithm 1, we estimate the trade-off curves for the considered NCS in two different scenarios: zero-input stability (see Remark \ref{rem:simplyProb}), and input-output stability for $\gamma \leq 5$. By selecting $T_s=0.01\, s$, the obtained trade-off curves are depicted in Fig. \ref{fig:Tradeoff}. It emerges that both zero-input and input-output stability are guaranteed up to  two consecutive packet dropouts. However, in case of input-output stability, the NCS allows a smaller value for $T_{mad}$ in the case of two consecutive packet dropouts. This is noticeable from the lower blue bar for $\Delta = 2$ in Fig. \ref{fig:Tradeoff}.

\begin{figure}[thpb]
\centering
\psfrag{x}[][][1]{$\Delta$}
\psfrag{y}[bc][][1]{$T_{mad} \, [s]$}
\includegraphics[scale=0.48]{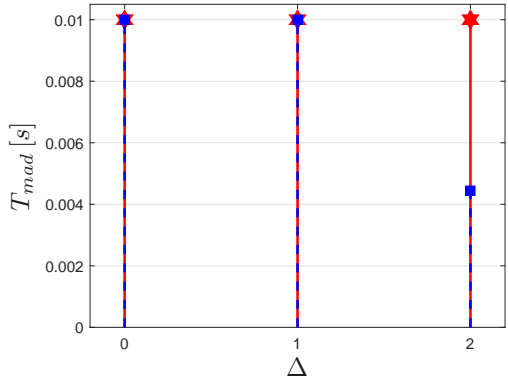}
\caption{Trade-off curves between $\Delta$ and $T_{mad}$ computed for the batch reactor by using Algorithm 1. In red trade-off curves in case of zero-input stability, while in blue trade-off curves in case of input-output stability with $\gamma \leq 5$.}
\label{fig:Tradeoff}
\end{figure}


\section{Conclusion}
In this letter, we presented a methodology to study stability of NCSs affected by varying transmission delays and DoS attacks. Relying on Lyapunov results for hybrid systems, an LMI based approach was devised to find explicit bounds for MANSD and MAD under which stability of a given NCS is preserved. The proposed approach was then applied to the well-known batch reactor. Trade-off curves between the above objectives were obtained by relying on the proposed algorithm. 

The extension of the methodology proposed in this letter to account for less conservative DoS models is currently part of our research.
\balance
\bibliographystyle{IEEEtran}
\bibliography{references}
\end{document}